\documentclass[runningheads]{llncs}

\usepackage{amsthm}
\usepackage[dvipdfmx]{graphicx}
\usepackage{amsfonts}
\usepackage{amsmath}
\usepackage{mathtools}
\usepackage[subrefformat=parens]{subcaption}
\usepackage{paralist}
\usepackage{algorithm}
\usepackage[noend]{algorithmic}
\usepackage{subcaption}
\usepackage{color}

\renewcommand{\algorithmiccomment}[1]{\bgroup\hfill//~#1\egroup}
  \newcommand{\argmax}{\mathop{\rm argmax}\limits}
 
\begin{document}
\title{Faster Algorithms for Evacuation Problems in Networks with the Single Sink of Small Degree \thanks{Supported by JSPS KAKENHI Grant Numbers 19H04068, 20K19746, 20H05794.}}
\titlerunning{Faster Algorithms for Evacuation Problems with the Small Degree Sink}
%
\author{
Yuya Higashikawa \and
Naoki Katoh \and
Junichi Teruyama \and
Yuki Tokuni
}

\authorrunning{Y.~Higashikawa~et~al.}
%
\institute{University of Hyogo, Kobe, Japan\\
\email{\{higashikawa,naoki.katoh,junichi.teruyama,ad21o040\}@gsis.u-hyogo.ac.jp}
}

\maketitle              
\begin{abstract}
In this paper, we propose new algorithms for {\it evacuation problems} defined on {\it dynamic flow networks}.
A dynamic flow network is a directed graph in which {\it source} nodes are given supplies (i.e., the number of evacuees) 
and a single {\it sink} node is given a demand (i.e., the maximum number of acceptable evacuees).
The evacuation problem seeks a dynamic flow that sends all supplies from sources to the sink  such that its demand is satisfied in the minimum feasible time horizon.
For this problem,
the current best algorithms are developed by Schl\"{o}ter~(2018) and Kamiyama~(2019), which run in strongly polynomial time but with high-order polynomial time complexity because they use submodular function minimization as a subroutine.  

In this paper, we propose new algorithms that do not explicitly execute submodular function minimization,
and we prove that they are faster than those by Schl\"{o}ter~(2018) and Kamiyama~(2019)
when an input network is restricted such that the sink has a small in-degree and every edge has the same capacity.

\end{abstract}
\section{Introduction}
The network routing problem taking into account the movement of commodities over time is important from the viewpoint of evacuation planning. 
To model such movement over time,
Ford and Fulkerson~\cite{Ford1962} introduced a {\it dynamic flow network}, 
which is a directed graph in which {\it source} nodes are given supplies (i.e., the number of evacuees), 
{\it sink} nodes are given a demand (i.e., the maximum number of acceptable evacuees), and
edges are given capacities and transit times. 
Here, the capacity of an edge bounds the rate at which flow can enter the edge per unit time, 
and the transit time represents the time required for evacuees to travel across the edge.

The {\it evacuation problem} is one of the most basic problems defined on a dynamic flow network.
In this problem, given a dynamic flow network with multiple sources and a single sink,
the goal is to find a dynamic flow that sends all supplies from sources to the sink such that its demand is satisfied in the minimum feasible time horizon.
For the evacuation problem, 
Kamiyama~\cite{Kamiyama2019} and Schl\"{o}ter~\cite{Schloter2018} independently proposed the current best $\tilde{O}(m^2k^5 + m^2 n k)$ time algorithms, where $n$ is the number of nodes, $m$ is the number of edges, and $k$ is the number of sources. 
In particular, it is known that the minimum feasible time horizon can be calculated in $\tilde{O}(m^2k^4 + m^2 n k)$ time~\cite{Schloter2018}.
Both the algorithms by \cite{Kamiyama2019,Schloter2018} use submodular function minimization as a subroutine,
which results in high-order polynomial time complexities.

However, 
for some restricted classes of networks,
there are efficient algorithms that do not call submodular function minimization algorithms~\cite{Kamiyama2006,Mamada2006}.
Mamada et~al.~\cite{Mamada2006} proposed an $O(n \log^2 n)$ time algorithm for tree networks, %
and Kamiyama et~al.~\cite{Kamiyama2006} proposed an $O(n\log n)$ time algorithm for restricted grid networks in which the edge capacities and transit times on edges are uniform and the edges are oriented such that supplies can take only the shortest paths to a single sink. 
However, for evacuation problems in more general classes of networks, all known strongly polynomial time algorithms depend on submodular function minimization.

{\bf Our contribution:}
In this paper, we propose a new algorithmic framework for evacuation problems that does not depend on submodular function minimization.
In particular, 
we present an efficient algorithm for networks in which the sink has a small in-degree and every edge has the same capacity.
Our main theorem is stated below.
\begin{theorem}\label{theo:cal_min_flow}
Given a dynamic flow network ${\cal N}$ with uniform edge capacities and a supply/demand function $w$, 
the evacuation problem can be solved in $\tilde{O}(m n d k^{d} + m^2k^2)$ time,
where $n$ is the number of nodes, $m$ is the number of edges, 
$d$ is the in-degree of the sink, and $k$ is the number of sources. 
\end{theorem}

\noindent
Note that when $d \le 5$, our algorithm runs in $\tilde{O}(m n k^{5}+ m^2 k^2)$ time, which is faster than the algorithms by \cite{Kamiyama2019,Schloter2018}.
For real-world applications, because many regions have a grid-pattern road network, the degree of almost all intersections in the road network
is at most 4~\cite{Stanislawski1947};
thus, the assumption of a sink with a small degree is natural.

In addition, we propose an even faster algorithm for grid networks in settings more general than those in~\cite{Kamiyama2006}, 
where $k=O(n)$ and
for every pair of adjacent nodes $u$ and $v$, there are two edges $(u,v)$ and $(v,u)$ with the same transit time and capacity.
While we immediately have an $\tilde{O}(n^6)$ time algorithm for this case by applying Theorem~\ref{theo:cal_min_flow} (because $m = O(n)$ and $d \le 4$ in a grid network),
we obtain a faster $\tilde{O}(n^4)$ time algorithm based on special properties of grid networks.

{\bf Main ideas and techniques:}
To solve the evacuation problem, the basic approach that we also employ is to compute the minimum feasible time horizon first and then find the corresponding dynamic flow, called the quickest flow.
For the former phase, \cite{Fleischer1998} showed that a trivial solution space is $O(2^k)$ in size, and for the latter phase, Schl\"{o}ter and Skutella~\cite{Skutella2017} proved that the quickest flow is obtained via a convex combination of at most $k$ vertices of a base polytope with $O(2^k)$ facets defined on some submodular function.
Our technical advancement is in proving a solution space of $O(k^d)$ size for computing the minimum feasible time horizon,
and the base polytope for the quickest flow also has $O(k^d)$ facets.
Thanks to the small number of facets of the base polytope, we also propose a method for efficiently determining vertices and corresponding coefficients one by one for a convex combination representation of the quickest flow.

{\bf Organization:}
This paper is organized as follows.
In Section~\ref{sec:preliminary}, we introduce notations and important concepts used throughout the paper.
In Section~\ref{sec:cal_opt_T}, we propose an algorithm to compute the minimum feasible time horizon.
In Section~\ref{sec:cal_flow}, we propose an algorithm to find the dynamic flow that sends all supplies from sources to the sink within the minimum feasible time horizon. 
In Section~\ref{sec:grid_net}, we describe a more efficient algorithm for grid networks. 
In Section~\ref{sec:quickest_p}, we conclude by discussing the application of our algorithm to the {\it quickest transshipment problem}, 
which is a generalization of the evacuation problem defined on dynamic flow networks with multiple sinks.

\section{Preliminaries}\label{sec:preliminary} 

\subsection{Notations and Problem Definition}\label{subsec:notations} 

Let $\mathbb{R}_+$ denote the set of positive real values.
A dynamic flow network ${\cal N}$ is given as a 5-tuple
${\cal N} = (D = (V, E), u, \tau, S^+,  S^-)$, 
where $D=(V, E)$ is a directed graph with node set $V$ and 
edge set $E$, 
$u$ is a capacity function $u: E \rightarrow \mathbb{R}_+$, 
$\tau$ is a transit time function $\tau: E \rightarrow \mathbb{R}_+$, 
and $S^+ \subseteq V$ and $S^- \subseteq V$ are sets of sources and sinks, respectively.
For a path $P$ on ${\cal N}$, let $|P|$ denote the total transit time along $P$, i.e., $|P| = \sum_{e \in P} \tau(e)$.
Throughout the paper, $u$ is a constant function, so we abuse $u$ as the constant capacity of every edge.
Furthermore, in the evacuation problem, $\mathcal{N}$ contains a single sink denoted by $s^-$ (i.e., $S^-=\{s^- \}$).
For each node $v \in V$, let $\delta^{+}(v)$ and $\delta^{-}(v)$ denote the set of 
out-going edges from $v$ and in-coming edges to $v$,
respectively. 
Let $|\delta^{+}(v)|$ and $|\delta^{-}(v)|$ be the out-degree and in-degree of node $v$, respectively. 
We use $n$, $m$, and $k$ as the cardinalities of $V$, $E$, and $S^+$, respectively.

As an input for the evacuation problem, 
we are given a dynamic flow network $\mathcal{N}$ and a supply/demand function $w: V \rightarrow \mathbb{R}$ that represents the amount of supply/demand at the sources/sinks. 
The value of the function $w$ is such that $w(v) > 0$ for $v \in S^+$, $w(v) < 0$ for $v \in S^-$ (i.e., for $v = s^{-}$), and $w(v) = 0$ for $v \in V \setminus (S^{+}\cup S^{-})$, and %
$\sum_{ v \in V } w(v)=0$. 
For any terminal subset $A \subseteq S^{+}\cup S^-$, we define $w(A)\coloneqq \sum_{s\in A} w(s)$. 

On a dynamic flow network $\mathcal{N}$, a {\it dynamic flow} $f$ is defined as a function $f: E \times \mathbb{R}_+ \rightarrow \mathbb{R}_+$, where $f(e, \theta)$ represents the flow rate entering edge $e$ at time $\theta$.
Let us consider the following constraints for a dynamic flow $f$:
\begin{equation}\label{eq:capacity_const}
0 \leq f(e,\theta) \leq u \quad \text{for each } e \in E, \text{ for each } \theta \in [0,\infty),
\end{equation}
\begin{equation}\label{eq:conserve_const}
\begin{split}
\int^{\theta}_{0}\left(\sum_{e\in \delta^{+}(v)}f(e,t)-\sum_{e\in \delta^{-}(v)}f(e,t-\tau(e)) \right)dt &\leq \max\{w(v),0\} \\
 & \text{for each } v \in V, \text{ for each } \theta \in [0,\infty).
\end{split}
\end{equation}
The constraints~\eqref{eq:capacity_const} and \eqref{eq:conserve_const} are called the {\it capacity constraint} and the {\it conserve constraint}, respectively. 
The conserve constraint~\eqref{eq:conserve_const} means that for any time $\theta$ and any node $v$, 
the amount of flow out of $v$ within time $\theta$ is at most the amount of flow entering $v$ within $\theta$ plus the amount of the supply at $v$.
Furthermore, for some $T \in \mathbb{R}_+$, consider the following constraint (known as the {\it supply/demand constraint}):
\begin{equation}\label{eq:demand_const}
\int^{T}_{0}\left(\sum_{e\in \delta^{+}(v)}f(e,t)-\sum_{e\in \delta^{-}(v)}f(e,t-\tau(e)) \right)dt = w(v) \text{ for each } v \in V.
\end{equation}
The supply/demand constraint~\eqref{eq:demand_const} implies that for each node $v \in V$, the net amount of flow
accumulated at $v$ within time $T$ equals its supply or demand. 
If a dynamic flow $f$ on $\mathcal{N}$ satisfies the above constraints~\eqref{eq:capacity_const}, \eqref{eq:conserve_const}, and \eqref{eq:demand_const}, then
$f$ is said to be \textit{feasible} w.r.t.\ $(w,T)$,
and $T$ is called a \textit{feasible time horizon} w.r.t.\ $(\mathcal{N},w)$.
Throughout the paper, $T^*$ denotes the minimum feasible time horizon w.r.t.\ $(\mathcal{N},w)$.
We call $f^*$ is the \textit{quickest flow} if $f^*$ is a feasible dynamic flow on $\mathcal{N}$ w.r.t.\ $(w,T^*)$. 
Then the evacuation problem requires computing $T^*$ and finding a corresponding dynamic flow.
The problem is described precisely as follows.\\

\noindent
{\sc Evacuation Problem}
\begin{itemize}
\item[{\bf Input:}] A dynamic flow network ${\cal N} = (D = (V, E), u, \tau, S^+,  S^-=\{s^-\})$ and a supply/demand function $w$.
\item[{\bf Goal:}] Find a solution $f$ of the problem
\begin{alignat*}{2}
{\rm minimize}      & \quad & & T                                                               \\
{\rm subject \ to}  &       & & f \text{ is a feasible dynamic flow on } \mathcal{N} \text{ w.r.t. } (w,T). 
\end{alignat*}
\end{itemize}

\subsection{Maximum Dynamic Flows}\label{subsec:MaxDynamicFlow}

Given a dynamic flow network $\mathcal{N}$, %
for a subset of the sources and sinks $A\subseteq S^+\cup S^-$ and a time horizon $\theta \in \mathbb{R}_+$, 
let $o^\theta (A)$ be the maximum amount of flow that can reach the sinks in $S^-\backslash A $ from the sources in $A$ within time horizon $\theta$. 
Note that $o^{\theta}(S^+\cup S^-)=0$ always holds.
We call the flow corresponding to $o^\theta (A)$ the \textit{maximum dynamic flow} from $A$ within $\theta$.
We assume that there is no restriction on the amount of flow out of each source of $A$. 
The following theorem by Fleischer and Tardos~\cite{Fleischer1998} %
gives a property of a feasible time horizon
\footnote{
Klinz~\cite{Klinz1994} showed the same property of Theorem~\ref{theo:Klinz} on a discrete-time model, and using this property, Fleischer and Tardos~\cite{Fleischer1998} showed Theorem~\ref{theo:Klinz} on a continuous-time model. Moreover, the original statement of Theorem~\ref{theo:Klinz}~\cite{Klinz1994} gives a similar property for a dynamic network ${\cal N}$ containing multiple sinks.
 }.

 \begin{theorem}[\cite{Fleischer1998}]\label{theo:Klinz}
 For a dynamic flow network ${\cal N}$, a supply/demand function $w$, and a time horizon $T \in \mathbb{R}_+$, 
 there exists a feasible dynamic flow on $\mathcal{N}$ w.r.t.\ $(w,T)$ if and only if 
 \begin{equation}\label{eq:min_oTA_wA_0}
 \min\{o^T(A)-w(A)\mid A\subseteq S^+\}\geq 0.
 \end{equation}
 \end{theorem}

For a subset of sources $A \subseteq S^+$, 
we define the \textit{minimum required time} $\theta(A)$ as 
\begin{equation}\label{eq:thetaA}
\theta(A) \coloneqq \min \{ \theta \mid o^{\theta}(A)-w(A) \ge 0 \}.
\end{equation}
Note that as shown later in \eqref{eq:oTA}, 
$o^T(A)$ is a nondecreasing continuous function in $T$ of range $[0, \infty)$.
Therefore,  by Theorem~\ref{theo:Klinz},
we immediately have the following corollary. 

\begin{corollary}\label{coro:T_theta}
For a dynamic flow network $\mathcal{N}$ and a supply/demand function $w$, 
it holds that 
\begin{equation*}\label{eq:min_T_max_theta}
T^*=\max\{\theta(A) \mid A \subseteq S^+ \}.
\end{equation*}
\end{corollary}

In the rest of this section, we give the properties of $o^T(A)$ and $\theta(A)$. 
For this purpose, we need to deal with a static flow network,
which corresponds to an input dynamic flow network $\mathcal{N}$.
A {\it static flow network} $\mathcal{\overline{N}}$ is a directed graph $D=(V,E)$ with a capacity $u(e)$ and a cost $c(e)$ for each edge $e \in E$.
On a static flow network $\overline{\mathcal{N}}$, 
a {\it static flow} $\bar{f}$ is defined as a function $\bar{f}: E \rightarrow \mathbb{R}_+$, where $\bar{f}(e)$ represents the amount of flow on edge $e$.
Given a source $s^+ \in V$ and a sink $s^- \in V$, 
a static flow $\bar{f}$ is said to be feasible if it holds that 
\begin{align}
    0 \leq \bar{f}(e) \leq u(e) & \text{ \hspace{20pt} for each }  e \in E, \label{eq:static_capacity_const}\\
    \sum_{e\in \delta^{+}(v)}\bar{f}(e)-\sum_{e\in \delta^{-}(v)}\bar{f}(e) =0 & \text{ \hspace{20pt} for each } v \in V \setminus\{s^+,s^-\}, \label{eq:static_conserve_const}\\
    \sum_{e\in \delta^{+}(s^-)}\bar{f}(e)-\sum_{e\in \delta^{-}(s^+)}\bar{f}(e) =0. & \label{eq:static_amount_const}
\end{align}
The {\it minimum-cost flow} %
is a feasible static flow %
$\bar{f}$ that minimizes $\sum_{e \in E} c(e) \bar{f}(e)$.
Given a static flow network $\overline{\mathcal{N}}$ and a feasible static flow %
$\bar{f}$ on $\overline{\mathcal{N}}$,
the residual network $\overline{\mathcal{N}}_{\bar{f}}$ is constructed as follows.
For every edge $e=(v_1,v_2) \in E$ such that $\bar{f}(e) > 0$, 
reduce its capacity to $u(e)-\bar{f}(e)$
and add new edge $e'=(v_2,v_1)$ of capacity $\bar{f}(e)$ and cost $-c(e)$.
In the following, given a dynamic flow network $\mathcal{N}$, 
$\overline{\mathcal{N}}$ means the static flow network comprising the same underlying graph as ${\mathcal{N}}$ with uniform edge capacity $u$ and cost $\tau(e)$ for every edge $e \in E$.

According to Anderson and Philpott~\cite{Anderson1994},
given $A\subseteq S^+\cup S^-$ and $T\in\mathbb{R}_+$, 
$o^T(A)$ can be obtained by applying the successively shortest path algorithm~\cite{Busacker1961,Iri1960,Jewell1958} for the minimum-cost flow problem in the following manner.

Initially, set $\bar{f}$ as a zero static flow, that is, $\bar{f}: E \rightarrow 0$.
At each step $i \ (\ge 1)$, execute the following two procedures.
\begin{itemize}
\item[(i)] 
Find the shortest (i.e., minimum cost) path $P^A_i$ from $A$ to $s^-$ in the current residual network $\overline{\mathcal{N}}_{\bar{f}}$.
If there is no such path, then break the iteration. 
\item[(ii)] 
Add to $\bar{f}$ a static flow of amount $u$ along path $P^A_i$. 
\end{itemize}

Let $p^A$ denote the number of paths obtained when the above iteration halts.
Algorithm~\ref{algo:successive_path} describes the above operations. 
\begin{algorithm}[bt]
\caption{{\sc SuccessiveShortestPath}}\label{algo:successive_path}
\begin{algorithmic}[1] 
	\REQUIRE $\overline{{\cal N}}$, $w$, $A$
	\ENSURE $p^A$, $(P^A_1, P^A_2, \ldots, P^A_{p^A})$
	\STATE $i \leftarrow 1$
	\STATE Let $\bar{f}$ be a zero static flow.
	\WHILE{\TRUE}
        \STATE Find the minimum cost path $P^A_i$ from $A$ to $s^-$ in the residual network $\overline{\mathcal{N}}_{\bar{f}}$.\\
        \COMMENT{If there are multiple paths, then choose $P^A_i$ by any tie-breaking rule to obtain the unique output.}
        \STATE If there is no such path, then break the while loop.
        \STATE Add to $\bar{f}$ a static flow of amount $u$ along path $P^A_i$.
        \STATE $i \leftarrow i + 1$
	\ENDWHILE
	\STATE $p^A \leftarrow i - 1$
	\RETURN $p^A$ and $(P^A_1, P^A_2, \ldots, P^A_{p^A})$.
\end{algorithmic}
\end{algorithm}

Then $o^T(A)$ is a piecewise linear function in $T$ (see Fig.~\ref{fig:oTA_T}) represented as
\begin{equation}\label{eq:oTA}
o^T(A)=\max_{h=1,\ldots,p^A}\left\{\sum^h_{i=1}(T-|P_i^A|)u, 0\right\},
\end{equation}
where $|P^A_i|$ denotes the sum of costs of the edges used in path $P^A_i$. %
\begin{figure}[t]
\centering
\includegraphics[width=60mm]{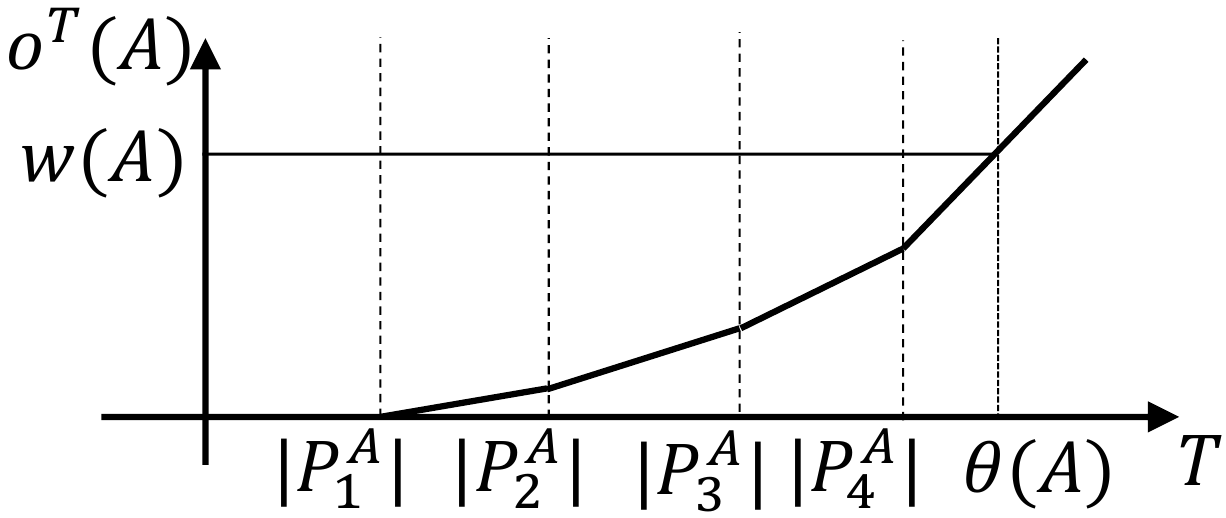}
\caption{$o^T(A)$ and the minimum required time $\theta(A)$ for $T$.
}\label{fig:oTA_T}
\end{figure}

As for $\theta(A)$,
because $o^T(A)$ is a nondecreasing continuous function in $T$ of range $[0, \infty)$ as shown in \eqref{eq:oTA},
$\theta(A)$ is $\theta$ satisfying $o^{\theta}(A)=w(A)$ by definition \eqref{eq:thetaA}.
We thus have
\begin{equation}\label{eq:min_need}
\theta(A)=\min_{h=1,\ldots,p^A} \left\{\frac{\sum^{h}_{i=1}|P^A_i|}{h}+\frac{w(A)}{hu}\right\}. 
\end{equation}

\section{ Computing the Minimum Feasible Time Horizon}\label{sec:cal_opt_T}

In this section, we propose an algorithm that computes the minimum feasible time horizon $T^*$ given the dynamic network ${\cal N}$ and a supply/demand function $w$. 
Throughout Sections~\ref{sec:cal_opt_T} and \ref{sec:cal_flow}, let $d$ denote the in-degree of the sink $s^-$.

Corollary~\ref{eq:min_T_max_theta} implies that 
the minimum feasible time horizon can be obtained by comparing $O(2^k)$ values of $\theta(A)$ for each $A \subseteq S^+$. 
In this section, we define a family of $O(k^d)$ subsets of $S^+$, denoted by $\hat{\mathcal{A}}$, 
and show that $\hat{\mathcal{A}}$ contains $A^* \subseteq S^+$ that maximizes $\theta(A)$. 

\subsection{Definition and Property of $\hat{\mathcal{A}}$}\label{sec:hatA}
We now consider a tuple of $p$ sources $(v_1,\ldots, v_p)\in (S^{+})^p$ with an integer $p \in \{ 1,\ldots,d \}$. 
We say that a source subset $A$ {\it admits} $(v_1,\ldots, v_p)$ if the following are satisfied:
(i) $p^A = p$; 
(ii) the origins of paths $P^A_1, \ldots, P^A_p$ are $v_1, \ldots, v_p$, respectively. 
Fig.~\ref{fig:exp_path_pA} shows the example of sources that admit given $(v_1,\ldots, v_p)$. 
We notice that for some $(v_1,\ldots, v_p)$, there may not exist any admitting $A \subseteq S^+$. 
 \begin{figure}[htbp]
   \begin{minipage}[b]{0.5\linewidth}
     \centering
     \includegraphics[width=50mm]{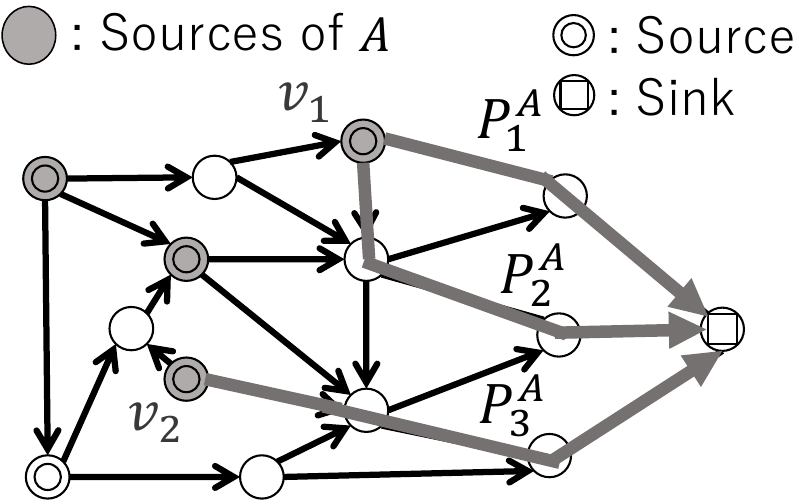}
     \subcaption{A source subset $A$ consisting of four sources (gray dots)}\label{fig:exp_path_pA_a}
   \end{minipage}
   \hspace{5pt}
   \begin{minipage}[b]{0.5\linewidth}
     \centering
     \includegraphics[width=50mm]{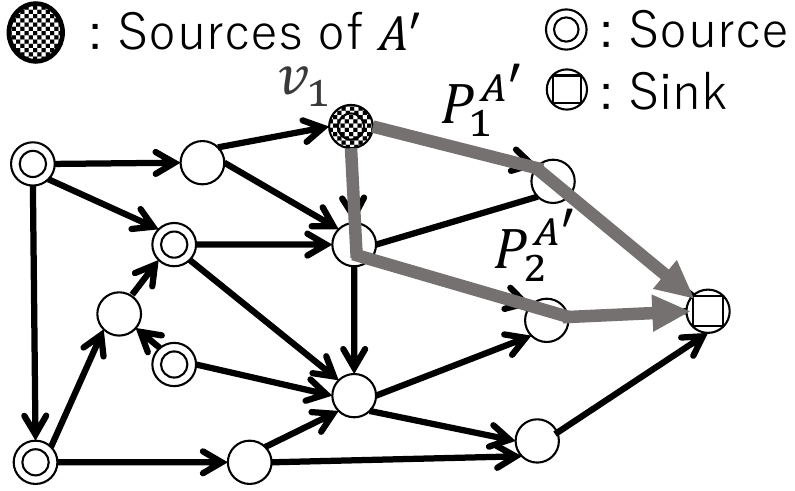}
     \subcaption{A source subset $A'$ consisting of one source (gray dot)}\label{fig:exp_path_pA_b}
  \end{minipage}
  \caption{Illustration of $P^A_i$ and $p^A$. In Fig.~\ref{fig:exp_path_pA}\subref{fig:exp_path_pA_a}, $p^A=3$ holds and $A$ admits $(v_1,v_1,v_2)$. In Fig.~\ref{fig:exp_path_pA}\subref{fig:exp_path_pA_b}, $p^{A^{'}}=2$ holds and $A^{'}$ admits $(v_1,v_1)$. }\label{fig:exp_path_pA}
\end{figure}

\begin{lemma}\label{lemm:bigger_A}
Let source subsets $A$, $A' \subseteq S^{+}$ admit the same $(v_1, \ldots, v_{p}) \in (S^{+})^p$, and let $w(A) \leq w(A')$ hold. 
Then $\theta(A) \leq \theta(A')$ holds.
\end{lemma}
\begin{proof}
By the behavior of Algorithm~\ref{algo:successive_path}, 
$P^{A}_i = P^{A'}_i$ holds for any $i \in \{1, \ldots, p\}$
because $A$ and $A'$ admit the same $(v_1, \ldots, v_{p})$. 
Then, \eqref{eq:oTA} implies that for any $\theta \in \mathbb{R}_+$, $o^{\theta}(A)=o^{\theta}(A')$ holds. 
Because $o^{\theta}(A)$ is nondecreasing on $\theta$ and $w(A) \leq w(A')$, we have $\theta(A)\leq \theta(A')$. 
\end{proof}

Corollary~\ref{coro:T_theta} and Lemma~\ref{lemm:bigger_A} imply that 
if source subsets $A, A' \subseteq S^{+}$ admit $(v_1, \ldots, v_{p}) \in (S^{+})^p$ and $w(A) < w(A')$ holds, 
then $\theta(A)$ must not be $T^*$. 
For each $(v_1,\ldots,v_p)\in (S^{+})^p$, we define $\hat{A}_{(v_1,\ldots,v_p)}$ as the subset of sources $A \subseteq S^+$ such that the %
cardinality of $A$ is the largest among all $A$ admitting $(v_1, \ldots, v_p)$:
\begin{equation}\label{eq:hat_A_v_1_v_p}
    \hat{A}_{(v_1,\ldots,v_p)} \coloneqq \argmax\{ |A| \mid A \subseteq S^+, A \text{ admits } (v_1, \ldots, v_p) \}.
\end{equation}
Note that if there is no subset that admits $(v_1,\ldots,v_p)$, then 
we define $\hat{A}_{(v_1,\ldots,v_p)}$ as the empty set.
Because $w(A)\geq 0$ holds for all source subsets $A$, $\hat{A}_{(v_1,\ldots,v_p)}$ is the source subset such that the amount of supply of $A$ is the largest among all $A$ admitting $(v_1, \ldots, v_p)$.

For each $p \in \{1, \ldots, d\}$, 
let $\hat{\mathcal{A}}_p$ denote the set of $\hat{A}_{(v_1,\ldots,v_p)}$ for all $(v_1,\ldots, v_p)\in (S^{+})^p$,
and let $\hat{\mathcal{A}}$ denote the union of all $\hat{\mathcal{A}}_p$:
\begin{equation}
    \hat{\mathcal{A}} \coloneqq \bigcup_{p \in \{ 1,\ldots,d \}} \hat{\mathcal{A}}_p, \text{ where } 
    \hat{\mathcal{A}}_p \coloneqq \{ \hat{A}_{(v_1,\ldots,v_p)} \mid (v_1,\ldots, v_p)\in (S^{+})^p \}.
\end{equation}
We see that $\hat{\mathcal{A}}$ consists of $O(k^d)$ subsets of $S^+$ (because $\sum_{p=1}^d k^p < 2k^d$ for $k > 1$). 

By the above argument, 
we have $T^*= \theta\left(\hat{A}_{(v_1,\ldots,v_p)}\right)$
for some $p \in \{1, \ldots, d\}$ and $(v_1,\ldots,v_p)\in (S^+)^p$.
We then have the following lemma.
\begin{lemma}\label{lemm:opt_T_A_hat}
For a dynamic flow network $\mathcal{N}$ and a supply/demand function $w$, 
it holds that 
\begin{equation}\label{eq:opt_T_A_hat}
T^*=\max\{\theta(A) \mid A \in \hat{\mathcal{A}} \}.
\end{equation}
\end{lemma}

\subsection{Algorithm for Finding the Minimum Feasible Time Horizon}

\begin{algorithm}[tb]
\caption{Obtain $\hat{\mathcal{A}}$}\label{algo:enu_A_hat}
\begin{algorithmic}[1] 
	\REQUIRE ${\cal N}, w$
	\ENSURE $\hat{\mathcal{A}}$
	\STATE Set $\hat{\mathcal{A}} \leftarrow \{ \emptyset \}$
		\FORALL {$p\in\{1, \ldots, d\}$}\label{line:loop_cal_A_hat}
			\FORALL {$(v_1,\ldots,v_p)\in (S^{+})^p$}\label{line:loop_cal_A_hat_2}
						\STATE Compute $\hat{A}_{(v_1,\ldots,v_p)}$ by calling Algorithm~\ref{algo:cal_A_hat} and add it to $\hat{\mathcal{A}}$ \label{line:cal_A_hat}
			\ENDFOR
		\ENDFOR\label{line:loop_cal_A_hat_end}
	\RETURN $\hat{\mathcal{A}}$
\end{algorithmic}
\end{algorithm}

Our algorithm for finding the minimum feasible time horizon involves the following two steps.
\begin{description}
\item[Step~1.] Enumerate all subsets in $\hat{\mathcal{A}}$ by Algorithm~\ref{algo:enu_A_hat}; that is,
compute $\hat{A}_{(v_1,\ldots,v_p)}$ for each $p \in \{1, \ldots, d\}, (v_1,\ldots,v_p)\in (S^+)^p$.
\item[Step~2.] For each $A \in \hat{\mathcal{A}}$, 
calculate $\theta(\hat{A}_{(v_1,\ldots,v_p)})$
and output the maximum value among them. 
\end{description}

To measure the running time for computing $\hat{A}_{(v_1,\ldots,v_p)}$, we have the following lemma.
See Section~\ref{subsec:hat_A} for the proof.
\begin{lemma}\label{lemm:comp_A_hat_time}
Given the dynamic network ${\cal N}$, a supply/demand function $w$, and $p\in\{1, \ldots, d\}$ and $(v_1,\ldots,v_p)\in (S^+)^p$, 
one can compute a subset $\hat{A} := \hat{A}_{(v_1,\ldots,v_p)}$
and paths $P^{\hat{A}}_1$, \ldots, $P^{\hat{A}}_{p^{\hat{A}}}$
in $O(n m d )$ time.
\end{lemma}

Because $\hat{\mathcal{A}}$ consists of $O(k^d)$ subsets, 
Lemma~\ref{lemm:comp_A_hat_time} implies the following lemma
about the running time of Step~1.

\begin{lemma}\label{lemm:enu_A_hat}
Given the dynamic network ${\cal N}$ and a supply/demand function $w$, 
one can construct $\hat{\mathcal{A}}$ and 
paths $P^A_1, \ldots, P^A_{p^A}$ for $A \in \hat{\mathcal{A}}$
in $O(n m d k^{d})$ time.
\end{lemma}

In Step~2, to compute $\theta(A)$ for each $A \in \mathcal{\hat{A}}$, we need only to compute $w(A)$
and values $(\sum^{h}_{i=1}|P^A_i|)/{h}+{w(A)}/(hu)$ for $h \in \{ 1,\ldots,p^A \}$ 
[by \eqref{eq:min_need}]
because we already have all $|P^A_i|~(i=1,\ldots,p)$ in Step~1 by Lemma~\ref{lemm:enu_A_hat}. 
Because we can compute $w(A)$ in $O(n)$ time and all values $\sum^{h}_{i=1}|P^A_i|$ in $O(d)$ time, 
Step~2 requires $O(n)$ time for each $A$ 
and $O(nk^d)$ time in total.
Summarizing the above, we have the following theorem. %

\begin{theorem}\label{theo:cal_min_T_time}
Given the dynamic network ${\cal N}$ whose edge capacities are uniform and a supply/demand function $w$, 
the minimum feasible time horizon can be computed in $O(n m d k^{d})$ time.
\end{theorem}

\subsection{Algorithm for Computing $\hat{A}_{(v_1,\ldots,v_p)}$}\label{subsec:hat_A}

In this subsection, we give an algorithm for computing $\hat{A}_{(v_1,\ldots,v_p)}$ 
for each $(v_1, \ldots ,v_p)$, and the proof of Lemma~\ref{lemm:comp_A_hat_time}. 
Our algorithm (see Algorithm~\ref{algo:cal_A_hat}) involves two steps:
first, check whether there exists a source subset that admits $(v_1, \ldots, v_p)$
by applying Algorithm~\ref{algo:successive_path} to a specific input; 
if yes, then find $\hat{A}_{(v_1,\ldots,v_p)}$ with modified Algorithm~\ref{algo:successive_path}. 

\begin{algorithm}[tb]
\caption{Computing $\hat{A}_{(v_1,\ldots,v_p)}$}\label{algo:cal_A_hat}
\begin{algorithmic}[1] 
	\REQUIRE ${\cal N}$, $w$, $(v_1, \ldots ,v_p) \in (S^{+})^p$
	\ENSURE $\hat{A}_{(v_1,\ldots,v_p)}$
    \STATE Set $A' \leftarrow \{ v_1,\ldots, v_p \}$
    \STATE Check whether $A'$ admits $(v_1,\ldots, v_p)$ by calling Algorithm~\ref{algo:successive_path} with ${\cal N}$, $w$, $A'$\label{line:check_admit}
    \IF{$A'$ admits  $(v_1,\ldots, v_p)$}
    \STATE Set $\hat{A} \leftarrow S^+$ and $\bar{f} \leftarrow$ a zero static flow\label{line:start_hat}
	\FOR {$i=1$ to $p$}\label{line:loop_1_to_p_A_prime}
		\STATE Remove from $\hat{A}$ every source node from which $s^-$ is closer than from $v_i$ in $\overline{{\cal N}}_{\bar{f}}$\label{line:remove_in_loop}
		\STATE Compute $P^{A'}_i$ and $\bar{f}^{A'}_i$\label{line:get_path}
		\STATE Set $\bar{f} \leftarrow \bar{f} + \bar{f}^{A'}_i$
	\ENDFOR \label{line:endloop_1_to_p_A_prime}
	\STATE Remove from $\hat{A}$ every source node that can reach $s^-$ in $\overline{{\cal N}}_{\bar{f}}$ \label{line:remove_reachabel}
	\RETURN $\hat{A}$ and paths $|P^{A'}_1|$, $\ldots$, $|P^{A'}_p|$\label{line:output_A_dou_prime}\label{line:end_hat}
	\ELSE
        \RETURN $\emptyset$
	\ENDIF
\end{algorithmic}
\end{algorithm}

To check the existence of a source subset that admits $(v_1, \ldots ,v_p)$, the following lemma is useful. 
The proof is given later. 
\begin{lemma}\label{lemm:exist_A}
For each $p\in\{1, \ldots, d\}$ and $(v_1,\ldots,v_p)\in (S^+)^p$, 
there exists a source subset $A$ that admits $(v_1, \ldots,v_p)$ 
if and only if 
set $\{v_1,\ldots,v_p\}$ admits $(v_1, \ldots,v_p)$.
\end{lemma}

Fig.~\ref{fig:A_bar_A} shows 
an example of source subsets $A$ and $A' = \{v_1, \ldots ,v_p\}$ admitting $(v_1, \ldots ,v_p)$.
By Lemma~\ref{lemm:exist_A}, one can check whether there exists a source subset that admits $(v_1, \ldots, v_p)$ 
by checking that $A' = \{v_1, \ldots, v_p\}$ admits $(v_1, \ldots, v_p)$
at Line~\ref{line:check_admit} of Algorithm~\ref{algo:cal_A_hat} as follows:
run Algorithm~\ref{algo:successive_path} with input $A' = \{v_1, \ldots, v_p\}$; 
if the origins of paths $P^{A'}_1, \ldots, P^{A'}_{p^{A'}}$---which are the outputs of Algorithm~\ref{algo:successive_path}---are equivalent to $v_1, \ldots, v_p$, 
then we see that $A'$ admits $(v_1, \ldots, v_p)$, otherwise it does not. 

\begin{figure}[tb]
  \begin{minipage}[t]{0.5\linewidth}
    \centering
    \includegraphics[width=50mm]{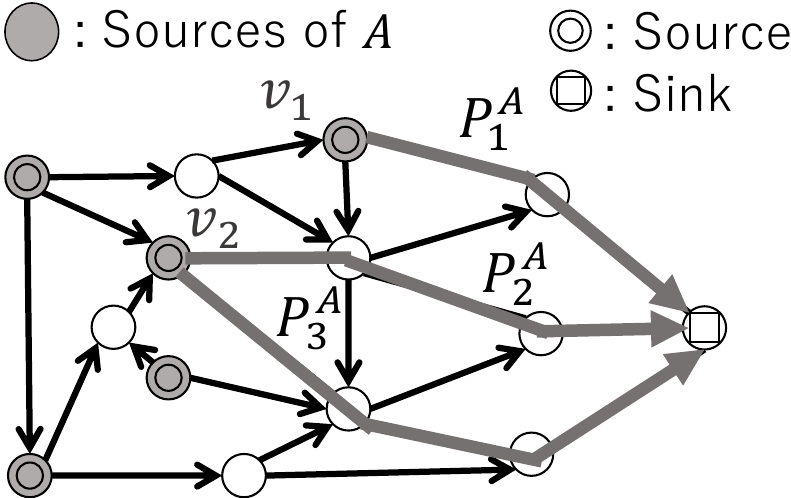}
    \subcaption{$A$ admits $(v_1,v_2,v_2)$}\label{fig:A_bar_a}
  \end{minipage}
  \hspace{5pt}
  \begin{minipage}[t]{0.5\linewidth}
    \centering
    \includegraphics[width=50mm]{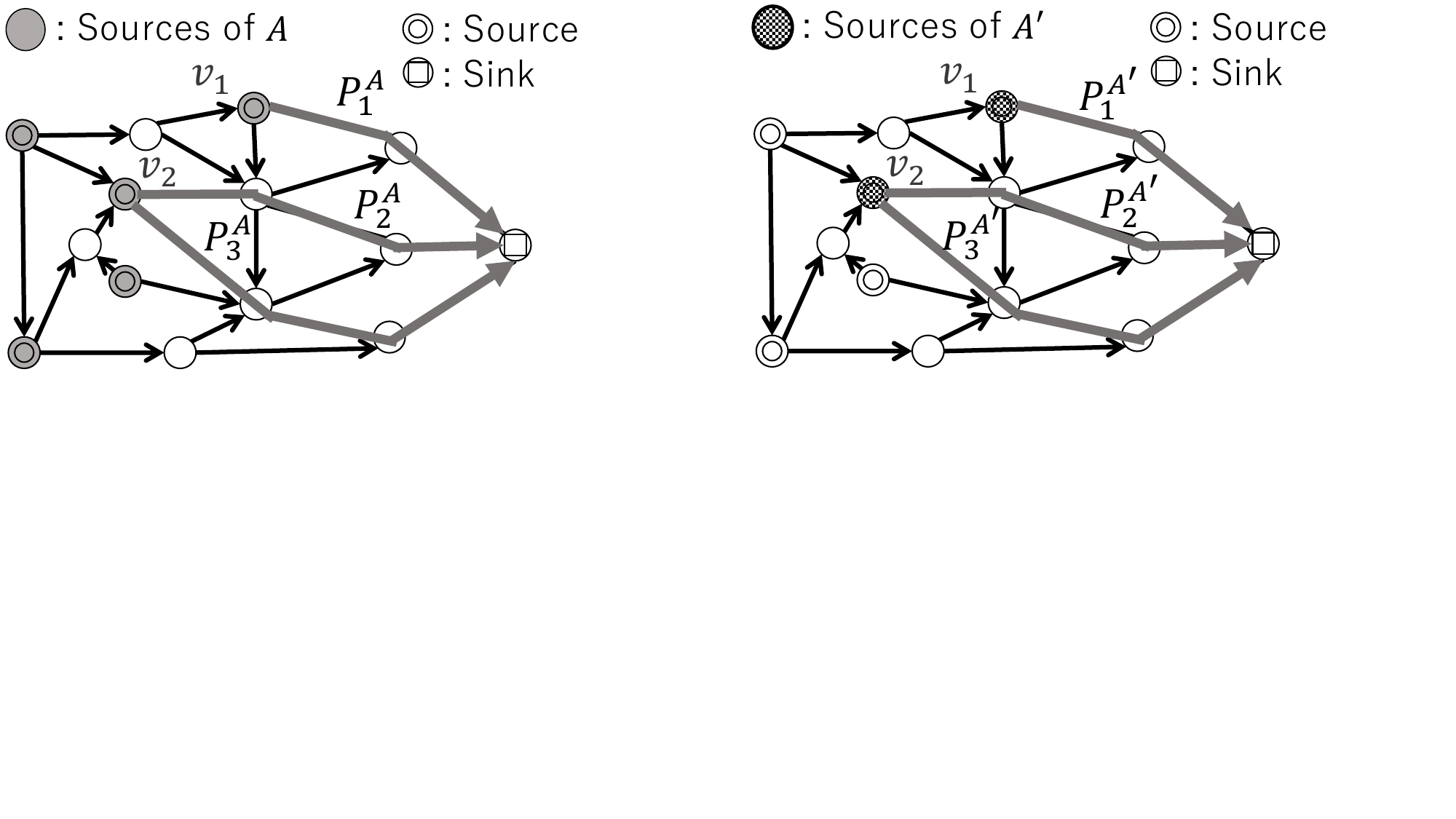}
    \subcaption{$A'=\{v_1,v_2\}$ admits $(v_1,v_2,v_2)$
    }\label{fig:A_bar_A_b}
  \end{minipage}
  \caption{Examples of the source subsets $A$ and $A'$ in Lemma~\ref{lemm:exist_A}. %
  }\label{fig:A_bar_A}
\end{figure}

Next, we describe operations to compute $\hat{A}_{(v_1,\ldots,v_p)}$ 
for given $(v_1,\ldots, v_p) \in (S^{+})^p$ such that $\{ v_1,\ldots, v_p \}$ admits $(v_1,\ldots, v_p)$. 
Before describing them, we give some notations and properties. 
We say that a source subset $A$ is a \textit{maximal} one admitting $(v_1, \ldots, v_p)$
if for any source $v \in S^+ \setminus A$, set $A \cup \{v\}$ does not admit $(v_1, \ldots, v_{p})$.
By the definition of $\hat{A}_{(v_1,\ldots,v_p)}$ in \eqref{eq:hat_A_v_1_v_p}, 
$\hat{A}_{(v_1,\ldots,v_p)}$ is a maximal source subset
admitting $(v_1, \ldots, v_p)$.
The following lemma implies that a maximal source subset admitting $(v_1, \ldots, v_p)$ is unique.
The proof is given later. 
\begin{lemma}\label{lemm:union}
Let $A_1$ and $A_2$ admit the same $(v_1, \ldots, v_p)$.
Then, $A_1 \cup A_2$ also admits $(v_1, \ldots, v_p)$.
\end{lemma}
By the uniqueness of the maximal source subset, 
it is enough to compute the maximal source subset admitting $(v_1,\ldots, v_p)$
to obtain $\hat{A}_{(v_1,\ldots,v_p)}$. 
See Lines~\ref{line:start_hat}--\ref{line:end_hat} of Algorithm~\ref{algo:cal_A_hat}.
Let us see which source node $\hat{A}_{(v_1,\ldots,v_p)}$ does not contain. 
Letting $A'$ be $ \{ v_1,\ldots, v_p \}$, for $i \in \{1,\ldots,p\}$, it holds that $P^{\hat{A}_{(v_1,\ldots,v_p)}}_i=P^{A'}_i$ and $\bar{f}^{\hat{A}_{(v_1,\ldots,v_p)}}_i=\bar{f}^{A'}_i$
because both $\hat{A}_{(v_1,\ldots,v_p)}$ and $A'$ admit $(v_1, \ldots ,v_p)$. 
Consider successively computing the minimum-cost paths $P^{A'}_1,\ldots,P^{A'}_p$ and 
corresponding static flows $\bar{f}^{A'}_1,\ldots,\bar{f}^{A'}_p$ in the aforementioned manner. 
Looking at the residual network $\overline{{\cal N}}_{\bar{f}}$ for $\bar{f} = \sum_{h=1}^{i-1} \bar{f}^{A'}_h$,
$\hat{A}_{(v_1,\ldots,v_p)}$ does not contain any source node from which $s^-$ is closer than from $v_i$ in $\overline{{\cal N}}_{\bar{f}}$;
otherwise, the origin of any minimum-cost path from $\hat{A}_{(v_1,\ldots,v_p)}$ to $s^-$ in $\overline{{\cal N}}_{\bar{f}}$ is never $v_i$, which is a contradiction.
At Line~\ref{line:remove_in_loop}, we remove such nodes from $\hat{A}$ which is a candidate for $\hat{A}_{(v_1,\ldots,v_p)}$. 
Furthermore, after taking all flows $\bar{f}^{A'}_1,\ldots,\bar{f}^{A'}_p$, there may remain source nodes that can reach $s^-$ in $\overline{{\cal N}}_{\bar{f}}$ for $\bar{f}=\sum_{h=1}^{p} \bar{f}^{A'}_h$ if $p < d$.
Then $\hat{A}_{(v_1,\ldots,v_p)}$ does not contain such nodes either.
At Line~\ref{line:remove_reachabel}, we remove such nodes and output the remaining source subset.
By the above operations, Algorithm~\ref{algo:cal_A_hat} computes a maximal subset admitting $(v_1, \ldots ,v_p)$, that is, $\hat{A}_{(v_1,\ldots,v_p)}$.

In the rest of this subsection, we prove Lemmas~\ref{lemm:comp_A_hat_time}, \ref{lemm:exist_A}, and \ref{lemm:union}.

\begin{proof}[Proof of Lemma~\ref{lemm:comp_A_hat_time}]
The correctness of Algorithm~\ref{algo:cal_A_hat} was discussed above, and here 
we analyze its running time. 

At Line~\ref{line:check_admit}, we call Algorithm~\ref{algo:successive_path}, which computes the shortest path from $A'$ to $s^-$ on a residual network calculated $d$ times.
We find the shortest path in $O(nm)$ time by the algorithm of Bellman and Ford~\cite{Bellman1958,Ford1962}.
Then, Line~\ref{line:check_admit} requires $O(nmd)$ time. 
At Line~\ref{line:get_path}, we also compute the shortest path from $A'$ to $s^-$ on a residual network.
Lines~\ref{line:loop_1_to_p_A_prime}--\ref{line:endloop_1_to_p_A_prime} are executed at most $d$ times and so require $O(nmd)$ time.
Lines~\ref{line:remove_reachabel} and~\ref{line:output_A_dou_prime} each require $O(n)$ time. 
In total, the running time of Algorithm~\ref{algo:cal_A_hat} is $O(n m d)$.
\end{proof}

\begin{proof}[Proof of Lemma~\ref{lemm:exist_A}]
Because the only-if part is obvious, we only need to prove the if part. 
Let us assume that a source subset $A$ admits $(v_1, \ldots ,v_p)$. 
By the behavior of Algorithm~\ref{algo:successive_path}, 
for any source $v \in A$ and any $i \in \{1,\ldots, p\}$, $|v_is^-| \leq |vs^-|$ holds in the residual network $\overline{{\cal N}}_{\bar{f}^A_{i-1}}$, 
where $|vs^-|$ denotes the minimum-cost flow from $v$ to $s^-$ in the static flow network.

Let $A'$ denote $\{v_1, \ldots, v_p\}$. We prove that $A'$ admits $(v_1, \ldots ,v_p)$ by induction. 
Because $A' \subseteq A$, $|v_1s^-| \leq |vs^-|$ holds for any source $v\in A'$ in the static network of $\overline{{\cal N}}$. 
Therefore, the origin of path $P_1^{A'}$ is $v_1$. 
Let us assume that for $i < p$, 
the origins of path $P_1^{A'}, \ldots, P_i^{A'}$
are $v_1, \ldots, v_{i}$, respectively. 
For each $j \in \{1, \ldots, i\}$, $P^{A'}_j=P^A_j$ holds, and thus we have $f^{A'}_{i} = f^A_{i}$. 
Because $A' \subseteq A$, $|v_{i+1}s^-|\leq |vs^-|$ holds for any source $v \in A'$ in the residual network $\overline{{\cal N}}_{\bar{f}^{A'}_{i}}$. 
Therefore, the origin of path $P_{i+1}^{A'}$ is $v_{i+1}$. 
Thus, the origins of paths $P_1^{A'}, \ldots, P_p^{A'}$ are $v_1, \ldots, v_{p}$, respectively. 
Moreover, for each $j \in \{1, \ldots, p\}$, $P^{A'}_j = P^A_j$ holds, and thus we have $\bar{f}^{A'}_{p}=\bar{f}^A_{p}$. 

The rest of the proof is to show that $p^{A'} = p$ holds. 
Because $A$ admits $(v_1, \ldots, v_p)$, there is no source in $A$ that can reach $s^-$ 
in the residual network $\overline{{\cal N}}_{\bar{f}^{A}_{p}}$. 
Because $\bar{f}^{A'}_p = f^A_p$ and $A' \subseteq A$ hold, 
there is no source in $A'$ that can reach $s^-$ in the residual network $\overline{{\cal N}}_{\bar{f}^{A'}_{p}}$. 
\end{proof}

\begin{proof}[Proof of Lemma~\ref{lemm:union}]
Because $A_1$ and $A_2$ admit the same $(v_1, \ldots, v_p)$, 
for each $i \in \{1,\ldots ,p\}$, $P^{A_1}_i=P^{A_2}_i$ and $\bar{f}^{A_1}_i=\bar{f}^{A_2}_i$ hold
by Algorithm~\ref{algo:successive_path}, which implies ${\cal \overline{N}}_{\bar{f}^{A_1}_{i}}={\cal \overline{N}}_{\bar{f}^{A_2}_{i}}$.
Because $A_1$ (resp.\ $A_2$) admits $(v_1, \ldots, v_p)$, 
for any $v \in A_1$ (resp.\ $A_2$) and for each $i \in \{1,\ldots, p\}$, 
$|v_is^-|\leq |vs^-|$ in residual network $\overline{\mathcal{N}}_{\bar{f}^{A_1}_{i-1}}(=\overline{\mathcal{N}}_{\bar{f}^{A_2}_{i-1}})$. 
Thus, for any $v \in A_1 \cup A_2$ and for each $i \in \{1,\ldots, p\}$, 
$|v_is^-|\leq |vs^-|$ in residual network $\overline{\mathcal{N}}_{\bar{f}^{A_1}_{i-1}}(=\overline{\mathcal{N}}_{\bar{f}^{A_2}_{i-1}})$.
Moreover, for any $v \in A_1 \cup A_2$, there is no path from $v$ to $s^-$ in $\overline{\mathcal{N}}_{\bar{f}^{A_1}_{p}}(=\overline{\mathcal{N}}_{\bar{f}^{A_2}_{p}})$
because $A_1$ and $A_2$ admit the same $(v_1, \ldots, v_p)$.
\end{proof}

\section{An Algorithm for Finding the Quickest Flow}\label{sec:cal_flow}
In this section, we propose an algorithm to find the quickest flow. 
Schl\"{o}ter and Skutella~\cite{Skutella2017} showed that a submodular function minimization algorithm relying on Cunningham's framework \cite{Cunningham1985} directly gives the quickest flow, and the current fastest algorithms~\cite{Kamiyama2019,Schloter2018} use this explicitly.
While our proposed algorithm does not execute any submodular function minimization explicitly,
it is based on the properties of the quickest flow as the solution for submodular function minimization, as shown by \cite{Skutella2017}.
In the following, we give basic definitions for submodular functions,
introduce the relationship between the quickest flow and submodular function minimization,
and then show our new algorithm for the quickest flow.

\subsection{Submodular Functions and Base Polytopes}

\begin{figure}[tb]
  \begin{minipage}[b]{0.5\linewidth}
    \centering
    \includegraphics[width=30mm]{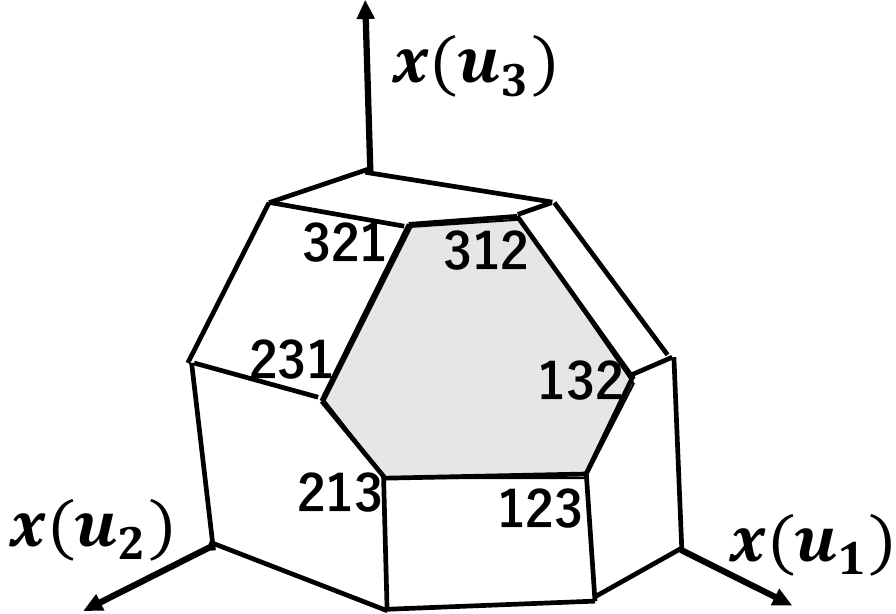}
    \subcaption{The gray face is a base polytope on $U=\{u_1,u_2, u_3\}$}\label{fig:base_poly_a}
  \end{minipage}
  \hspace{5pt}
  \begin{minipage}[b]{0.5\linewidth}
    \centering
    \includegraphics[width=30mm]{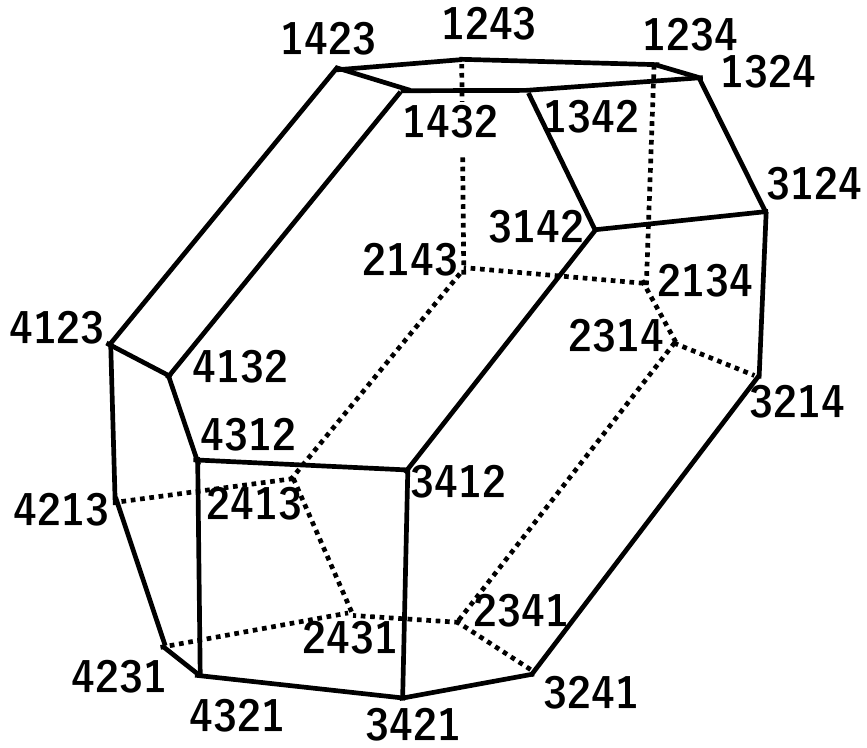}
    \subcaption{Example of base polytope on $U=\{u_1,u_2, u_3, u_4\}$ 
    }\label{fig:base_poly_b}
  \end{minipage}
  \caption{Base polytopes for submodular functions on $U$ with three or four elements. %
  }\label{fig:base_poly}
\end{figure}

Let $U$ be a nonempty finite set. A set function $g:2^U\rightarrow \mathbb{R}$ 
is said to be a \textit{submodular function} if it satisfies $g(X) + g(Y ) \geq g(X \cup Y ) + g(X \cap Y )$
for every pair of subsets $X, Y \subseteq U$. 
For a vector $x \in \mathbb{R}^U$ and a subset $A \subseteq U$,
we define $x(A) \coloneqq \sum_{s\in A} x(s)$. 
Given a submodular function $g$ defined on $U$, a \textit{base polytope} $B(g)$ is a convex polyhedron on $\mathbb{R}^U$ defined as
\begin{equation}\label{eq:base_polytope}
B(g) \coloneqq \left\{ x\in \mathbb{R}^U \ \Bigg| \ x(A)\leq g(A) \ \text{for all } A\subseteq U, \text{ and } x(U)= g(U) \right\}.
\end{equation}
Let us observe a property of the vertices of $B(g)$.
We can see that a vertex of $B(g)$ is the solution of simultaneous equations
consisting of $x(U)= g(U)$ and $x(A) = g(A)$ for other $|U|-1$ subsets $A \subset U$. 
Furthermore, each vertex of $B(g)$ corresponds to a total order on $U$ (see Fig.~\ref{fig:base_poly}). 
More precisely, given a total order $\prec \ = (u_{i_1}, \ldots, u_{i_{|U|}})$ on $U$, 
let $b^{(\prec,g)}$ denote
the solution $x$ of the following simultaneous equations:
\begin{equation}\label{eq:vertex_order_B}
x(\{u_{i_1}, \ldots, u_{i_l}\})=g(\{u_{i_1}, \ldots, u_{i_l}\}) \ \text{for each } l \in \{1,\ldots, |U|\}.
\end{equation}

\subsection{Relationship between Quickest Flows and Base Polytopes}
\label{sec:known_qf_bp}

Recall the definition of the function $o^T: S^+ \cup S^- \rightarrow \mathbb{R}$: 
for a subset $A \subseteq S^+ \cup S^-$, $o^T(A)$ represents the maximum amount of flow that can reach the sinks in $S^- \setminus A$ from the sources in $A$ within time horizon $T$. 
It is observed in~\cite{Hoppe2000} that 
$o^T$ is a submodular function for any time horizon $T\in \mathbb{R}_+$.
Note that in our case, $B(o^T)$ is defined on $\mathbb{R}^{k+1}$ because $|S^+|=k$ and $|S^-|=1$.
The following lemma is known for the relationship between the base polytope $B(o^T)$ and the feasibility of $({\cal N}, w, T)$. 
\begin{lemma}[\cite{Skutella2017}]\label{lemm:base_polytope_feasible}
Given a dynamic network ${\cal N}$, a supply/demand function $w$, and a time horizon $T$, 
there is a feasible flow w.r.t.\ $(w,T)$ on dynamic network $\mathcal{N}$ if and only if 
it holds that $w(s)_{s \in S^+ }\in B(o^T)$.
\end{lemma}

To see the relationship between the quickest flow and the vertices of the base polytope, we introduce the concept of \textit{lexicographically maximal dynamic flow} (\textit{lex-max dynamic flow} for short).
Given a dynamic flow network $\mathcal{N}$, a total order $\prec$ on $S^+ \cup S^-$,
and a time horizon $T\in \mathbb{R}_+$,
the lex-max dynamic flow w.r.t.\ $(\prec, T)$, denoted by $f^{(\prec,T)}$, 
is the dynamic flow on $\mathcal{N}$ satisfying 
the capacity constraint~\eqref{eq:capacity_const} and 
the conserve constraint~\eqref{eq:conserve_const}
that maximizes the amount of flow leaving the sources in order of $\prec$ within time horizon $T$. 

Schl\"{o}ter and Skutella~\cite{Skutella2017} showed that
for any total order $\prec$ on $S^+ \cup S^-$,
the vertex $b^{(\prec,o^T)}$ of $B(o^T)$ corresponds to $f^{(\prec,T)}$; in other words, for any source/sink $s \in S^+ \cup S^-$, $b^{(\prec,o^T)}(s)$ is the amount of flow leaving $s$ by lex-max dynamic flow $f^{(\prec,T)}$.
Moreover, considering a supply/demand function $w$ as a vector in $\mathbb{R}^{k+1}$, 
there are some total orders $\prec_1, \ldots, \prec_h$ on $S^+ \cup S^-$
such that $w$ can be a convex combination of $b^{(\prec_1,o^{T^*})}, \ldots, b^{(\prec_h,o^{T^*})}$.
That is, there are some total orders $\prec_1, \ldots, \prec_h$ and positive real values $\lambda_1, \ldots, \lambda_h$ with $h\leq k$ such that 
\begin{equation}\label{eq:conv_comb_w}
    w = \sum^{h}_{i=1} \lambda_i b^{(\prec_i,o^{T^*})}.
\end{equation}
Recall that $T^*$ is the minimum feasible time horizon w.r.t.\ $w$.
Then the quickest flow $f^*$ on $\mathcal{N}$ w.r.t.\ $w$ satisfies
\begin{equation}\label{eq:conv_comb_f}
    f^*=\sum^{h}_{i=1} \lambda_i f^{(\prec_i,T^*)}.
\end{equation}

In the rest of this subsection, we give the property of the base polytope $B(o^{T^*})$ with respect to $\hat{\mathcal{A}}$.
Let $\hat{B}(o^{T^*})$ denote a convex polyhedron on $\mathbb{R}^{k+1}$ defined as
\begin{equation*}
\hat{B}(o^{T^*}) \coloneqq \left\{ x\in \mathbb{R}^{k+1} \ \Bigg| \ x(A)\leq o^{T^*}(A) \ \text{for all } A \in \hat{\mathcal{A}}, \text{ and } x(S^+ \cup S^-)= 0 \right\}.
\end{equation*}
By definition, it clearly holds that $B(o^{T^*}) \subseteq \hat{B}(o^{T^*})$.
Not only that, we can show that $B(o^{T^*}) = \hat{B}(o^{T^*})$ as follows.
Consider any vector $x' \in \mathbb{R}^{k+1}$ on any facet of $B(o^{T^*})$.
Then the minimum feasible time horizon w.r.t.\ $x'$ is also $T^*$.
As discussed in Section~\ref{sec:hatA}, $\hat{\mathcal{A}}$ does not depend on a supply/demand function,
thus by Lemma~\ref{lemm:opt_T_A_hat},
$\hat{\mathcal{A}}$ contains the subset of the sources $A'$ that maximizes $\theta(A)$ w.r.t.\ $x'$.
This means that the facet of $B(o^{T^*})$ containing $x'$ is determined by 
\begin{equation*}
    x(A')= o^{T^*}(A') \text{ and } x(S^+ \cup S^-)=0,
\end{equation*}
which is also a facet of $\hat{B}(o^{T^*})$.
We thus have the following lemma.
\begin{lemma}\label{lemm:eqvlt_bp}
For a dynamic flow network $\mathcal{N}$ and a supply/demand function $w$, 
let $T^*$ be the minimum feasible time horizon $T^*$ w.r.t.\ $(\mathcal{N},w)$.
Then it holds that $B(o^{T^*}) = \hat{B}(o^{T^*})$.
\end{lemma}

\subsection{Algorithms}
From Section~\ref{sec:known_qf_bp}, our main task is to determine total orders $\prec_1, \ldots, \prec_h$ on $S^+ \cup S^-$ and positive real values $\lambda_1, \ldots, \lambda_h$ satisfying \eqref{eq:conv_comb_w}.
By Lemma~\ref{lemm:eqvlt_bp}, we can consider $\hat{B}(o^{T^*})$ instead of $B(o^{T^*})$.
For a dynamic flow network $\mathcal{N}$ and a supply/demand function $w$,
let $A^*$ denote the subset that maximizes $\theta(A)$, meaning that $T^* = \theta(A^*)$ holds.

The algorithm inductively computes $\prec_i,f^{(\prec_i,{T^*})},b^{(\prec_i,o^{T^*})},\lambda_i$ for $i = 1,\ldots,h$ in this order.
Below, we detail the first iteration; then, in the next paragraph we detail the $i$th iteration. %

Let $x'_1= w \in \mathbb{R}^{k+1}$ and $A_1 = A^*$.
Recall that for $x = x'_1$, it holds that
\begin{equation}\label{eq:triv_eqs}
    x(A_1)= o^{T^*}(A_1) \text{ and } x(S^+ \cup S^-)= o^{T^*}(S^+ \cup S^-)=0.
\end{equation}
In other words, $x'_1$ is located on the facet of $\hat{B}(o^{T^*})$ determined by the above two equations in \eqref{eq:triv_eqs}, denoted by $B_1$.
This implies that all the vertices $b^{(\prec_1,o^{T^*})},\ldots,b^{(\prec_h,o^{T^*})}$ satisfying \eqref{eq:conv_comb_w} are located around $B_1$.
We thus arbitrarily choose a total order $\prec_1 \ = (s_{i_{1,1}},\ldots,s_{i_{1,k+1}})$ on $S^+ \cup S^-$ so that $b^{(\prec_1,o^{T^*})}$ is one of the vertices of $B_1$, that is, by~\eqref{eq:vertex_order_B},
$\{ s_{i_{1,1}},\ldots,s_{i_{1,|A^*|}} \}$ and $\{ s_{i_{1,|A^*|+1}},\ldots,s_{i_{1,k+1}} \}$ coincide with $|A^*|$ and $(S^+ \cup S^-) \setminus |A^*|$, respectively.
We then compute the lex-max dynamic flow $f^{(\prec_1,T^*)}$ using an algorithm by Hoppe and Tardos~\cite{Hoppe2000},
which immediately gives $b^{(\prec_1,o^{T^*})}$ as mentioned in Section~\ref{sec:known_qf_bp}.
Next, to find the vertex $b^{(\prec_2,o^{T^*})}$, 
we determine the facet of $B_1$ denoted by $B_2$ with which the half line from $b^{(\prec_1,o^{T^*})}$ to $x'_1$ in $\mathbb{R}^{k+1}$ intersects. %
Recall that $B_2$ is determined by the two equations in \eqref{eq:triv_eqs} and one more equation $x(A')= o^{T^*}(A')$ for some $A' \in \hat{\mathcal{A}}\setminus\{A_1,S^+ \cup S^-\}$.
We thus calculate the intersection points on all $O(|\hat{\mathcal{A}}|)$ candidate hyperplanes with the half line, and choose one that minimizes the Euclidean distance from $b^{(\prec_1,o^{T^*})}$ to the intersection point.
Let $A_2$  be an element of $\hat{\mathcal{A}}$ corresponding to the chosen hyperplane, and let $x'_2$ be the intersection point on $B_2$.
Once we obtain $x'_2$, we can have $x'_1$ as a convex combination of $b^{(\prec_1,o^{T^*})}$ and $x'_2$ and then employ the coefficient of $b^{(\prec_1,o^{T^*})}$ as $\lambda_1$.
Subsequently, we will represent $x'_2$ as a convex combination of vertices around $B_2$ in an inductive manner.

After the $(i-1)$-th iteration, we can determine $b^{(\prec_{i},o^{T^*})}$, $B_{i}$, $A_1,\ldots,A_i$, and $x'_{i}$ in a similar way. 
To find the vertex $b^{(\prec_{i+1},o^{T^*})}$, 
we determine the facet of $B_{i}$ denoted by $B_{i+1}$ with which the half line from $b^{(\prec_{i},o^{T^*})}$ to $x'_{i}$ in $\mathbb{R}^{k+1}$ intersects.
Recall that $B_{i+1}$ is determined by 
\begin{equation}\label{eq:ith_ite}
x(A_h)=o^{T^*}(A_h)\ \text{for}\ h=\{1,\ldots, i\}\ \text{and}\  x(S^+\cup S^-)=0,
\end{equation}
and one more equation $x(A')= o^{T^*}(A')$ for some $A' \in \hat{\mathcal{A}}\setminus\{A_1,\ldots, A_{i},S^+ \cup S^-\}$.
We thus calculate the intersection points on all $O(|\hat{\mathcal{A}}|)$ candidate hyperplanes
with the half line, and choose one that minimizes the Euclidean distance from $b^{(\prec_i,o^{T^*})}$ to the intersection point.
Let $A_{i+1}$ be an element of $\hat{\mathcal{A}}$ corresponding to the chosen hyperplane, and let $x'_{i+1}$ be the intersection point on $B_{i+1}$.
By arbitrarily choosing a total order $\prec_{i+1} \ = (s_{I_{i,1}},\ldots,s_{I_{i,k+1}})$ on $S^+ \cup S^-$ satisfying $\{ s_{I_{i,1}},\ldots,s_{I_{i,|A_h|}} \}=A_h$ for $h=\{1,\ldots, i+1\}$, respectively, 
we can obtain $b^{(\prec_{i+1},o^{T^*})}$, which is one of the vertices of $B_{i+1}$.

The difference between the dimensions of $B_{i}$ and $B_{i+1}$ is at least one. Therefore, the iteration stops after $k$ times at most. 
Algorithm~\ref{algo:find_QF} gives a formal description. %

\begin{algorithm}[tb]
\caption{Find the quickest flow $f^*$}\label{algo:find_QF}
\begin{algorithmic}[1] 
	\REQUIRE ${\cal N}, w, T^*, \hat{\mathcal{A}},(P^A_1,\ldots,P^A_{p^A})$ for all $A\in \hat{\mathcal{A}}, A^*$
	\ENSURE $f^*$
	\STATE Set $i \leftarrow 1$, $x'_1 \leftarrow w$, $\alpha_1 \leftarrow 1$, and $\bar{A}_1 \leftarrow A^*$
	\WHILE {$x'_i$ is not a vertex of $\hat{B}(o^{T^*})$}\label{line:loop_x_not_vertex}
        \STATE Arbitrarily choose a total order $\prec_i$ on $S^+ \cup S^-$ so that the first $|\bar{A}_h|$ nodes coincide with the elements of $\bar{A}_h$ for all $h \in \{1,\ldots,i\}$
		\STATE Compute $f^{(\prec_i,{T^*})}$ and calculate $b^{(\prec_i,o^{T^*})}$
		\STATE Calculate the intersection points on hyperplanes $x(A)= o^{T^*}(A)$ for all $A \in \hat{\mathcal{A}}$ with the half line from $b^{(\prec_i,o^{T^*})}$ to $x'_i$, 
		set $A_{i+1}$ as the one that minimizes the distance from $b^{(\prec_i,o^{T^*})}$ to the intersection point, and
		set $x'_{i+1}$ as that intersection point
		\STATE Calculate positive real values $\beta_i$ and $\gamma_i$ satisfying $x'_i = \beta_i b^{(\prec_i,o^{T^*})} + \gamma_i x'_{i+1}$, and set $\lambda_i \leftarrow \alpha_i \beta_i$ and $\alpha_{i+1} \leftarrow \alpha_i \gamma_i$
		\STATE Search $j \in \{0,\ldots,i\}$ such that $\bar{A}_j \subseteq A_{i+1} \subseteq \bar{A}_{j+1}$ (where $\bar{A}_0=\emptyset$ and $\bar{A}_{i+1}=S^+ \cup S^-$), and set $\bar{A}_{j+1} \leftarrow A_{i+1}$ and $\bar{A}_{h+1} \leftarrow \bar{A}_{h}$ for $h \in \{j+1,\ldots,i\}$
		\STATE Set $i \leftarrow i+1$
	\ENDWHILE\label{line:loop_x_not_vertex_end}
	\RETURN  $\lambda_1 f^{(\prec_1,{T^*})}+ \cdots +\lambda_i f^{(\prec_i,{T^*})}$\label{line:output_f}
\end{algorithmic}
\end{algorithm}

Here, we discuss the running time of Algorithm~\ref{algo:find_QF}.
Line~3 takes $O(k^2)$ time.
Using an algorithm by Hoppe and Tardos~\cite{Hoppe2000} for computing a lex-max dynamic flow (together with Orlin's min-cost flow algorithm~\cite{Orlin1993} and Thorup's shortest path algorithm~\cite{Thorup2004}), 
Line~4 takes $O(mk(m + n \log \log n)\log n)$ time.
Line~5 takes $O(k|\hat{\mathcal{A}}|)=O(k^{d+1})$ time.
Line~6 takes $O(1)$ time if we reuse the results calculated in Line~5.
Line~7 takes $O(k)$ time.
We repeat Lines~3--8 at most $k$ times. 
Combining the above argument and Theorem~\ref{theo:cal_min_T_time}, we have the following main theorem, which 
repeats Theorem~\ref{theo:cal_min_flow} but with logarithmic factors.
\begin{theorem}\label{theo:cal_min_flow_2}
Given a dynamic flow network ${\cal N}$ with uniform edge capacities and a supply/demand function $w$, 
the evacuation problem can be solved in $O(m n d k^{d} + mk^2 (m + n \log \log n)\log n)$ time.
\end{theorem}

\section{More-efficient Algorithms for Grid Networks}\label{sec:grid_net}

In this section, we consider evacuation problems 
on bidirected grid networks with uniform edge capacity and transit time.
The {\it bidirected grid network} ${\cal G}$ is 
a dynamic flow network with a square grid graph $G = (V, E)$
such that $V$ consists of $n = N \times N$ nodes and 
each pair of adjacent nodes $u$ and $v$ are linked by edges $(u, v)$ and $(v, u)$
(see Fig.~\ref{fig:grid_net}\subref{fig:grid_net_undirected}). 
Let us assume that the capacity and transit time of all edges are uniform in ${\cal G}$, 
and the source set $S^+ = V \setminus \{s^-\}$.
Throughout this section, we abuse $\tau$ as the constant transit time of every edge
since $\tau$ is a constant function.

\begin{remark}
Kamiyama et~al.~\cite{Kamiyama2006} considered other grid graphs in 
which each pair of adjacent nodes $u$ and $v$ are linked by an edge $(u, v)$
if and only if $u$ is closer than $v$ to the sink $s^-$.
See Fig.~\ref{fig:grid_net}\subref{fig:grid_net_kamiyama}.
\end{remark}
\begin{figure}[htbp]
\begin{tabular}{cc}
  \begin{minipage}[t]{0.45\linewidth}
    \centering
    \includegraphics[width=2.4cm]{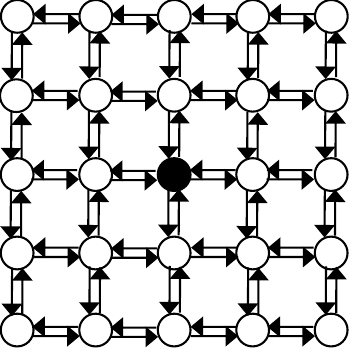}
    \subcaption{
    A bidirected grid network:
    for every pair of adjacent nodes $u$ and $v$, there are two edges $(u,v)$ and $(v,u)$ of the same transit time and capacity
    }\label{fig:grid_net_undirected}
    \end{minipage}
  & \hspace{5pt}
  \begin{minipage}[t]{0.45\linewidth}
    \centering
    \includegraphics[width=2.4cm]{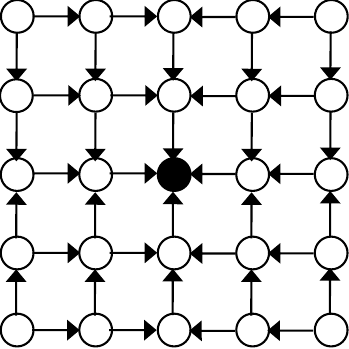}
    \subcaption{
    A grid network considered by Kamiyama et~al.~\cite{Kamiyama2006}: the edges are oriented so that supplies can only take the shortest path to the sink}\label{fig:grid_net_kamiyama}
  \end{minipage}
\end{tabular}
  \caption{
  Two types of grid network. The white and black circles represent sources and sinks, respectively.
  }\label{fig:grid_net}
\end{figure}

For bidirected grid network $\mathcal{G}$ with $n$ nodes, we have the number of edges $m = O(n)$,  the in-degree of the sink $d = 4$, and the number of sources $k = n-1$.
Then, Theorem~\ref{theo:cal_min_flow} implies that 
the evacuation problem on $\mathcal{G}$ can be solved in $O(n^6)$ time. 
To propose more-efficient algorithms for the evacuation problems on $\mathcal{G}$,
we give a stricter analysis of the existence of a subset that admits $(v_1,\ldots,v_p) \in (S^{+})^p$. 
Recall that $\hat{A}_{(v_1,\ldots,v_p)}$ is a subset $A \subseteq S^+$ such that 
$w(A)$ is the largest among all $A$ admitting $(v_1, \ldots, v_p)$ for each $p \in \{1,2,3,4\}$ and $(v_1,\ldots,v_p) \in (S^{+})^p$, 
and $\hat{\mathcal{A}}$ is the family of all $\hat{A}_{(v_1,\ldots,v_p)}$.
In Section~\ref{sebsec:hat_A_grid}, we show that the cardinality of $\hat{\mathcal{A}}$ is 
$O(n^2)$ instead of $O(n^4)$.
This fact implies the following theorem
by the same proof as that of Theorem~\ref{theo:cal_min_flow}.

\begin{theorem}\label{thm:time_grid}
Given a bidirected grid network ${\cal G}$ and a supply/demand function $w$, 
the evacuation problems can be solved in $\tilde{O}(n^4)$ time. 
\end{theorem}

\subsection{Analysis of the Cardinality of $\mathcal{\hat{A}}$}\label{sebsec:hat_A_grid}
In this subsection, for a bidirected grid network $\mathcal{G}$ with $n$ nodes, 
we show that $|\hat{\mathcal{A}}|$ is $O(n^2)$.
We give a necessary condition for the existence of source subsets admitting
$(v_1, \ldots, v_p)$ for given $p \in \{1,2,3,4\}$ and $(v_1, \ldots, v_p)$.
Fig.~\ref{fig:not_exist_A} shows an example in which there is no subset admitting some $(v_1,\ldots ,v_p)$.

\begin{figure}[tb]
\begin{tabular}{cc}
  \begin{minipage}[t]{0.45\linewidth}
    \centering
    \includegraphics[width=40mm]{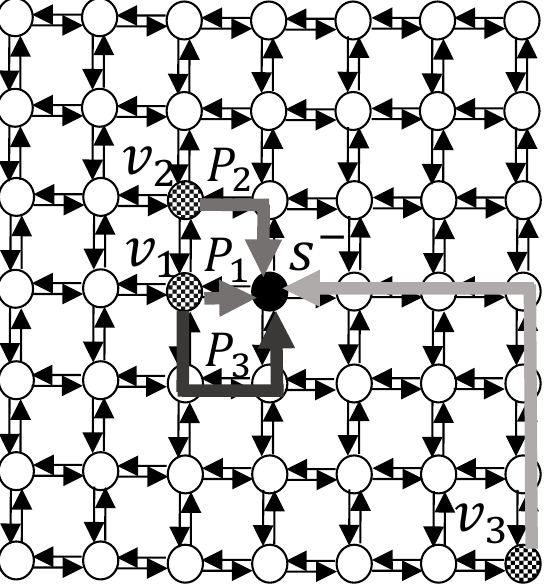}
    \caption{Example of $(v_1,\ldots ,v_p)$ such that there is no source subset $A$ admitting
    $(v_1,\ldots ,v_p)$. 
    In this case, there is no source subset $A$ for which the origin of $P^A_3$ is $v_3$
    because a path $P_3$ from $v_1$ to $s^-$ is shorter than any path from $v_3$ to $s^-$.  
    }\label{fig:not_exist_A}
    \end{minipage}
  & \hspace{5pt}
  \begin{minipage}[t]{0.45\linewidth}
    \centering
        \includegraphics[width=45mm]{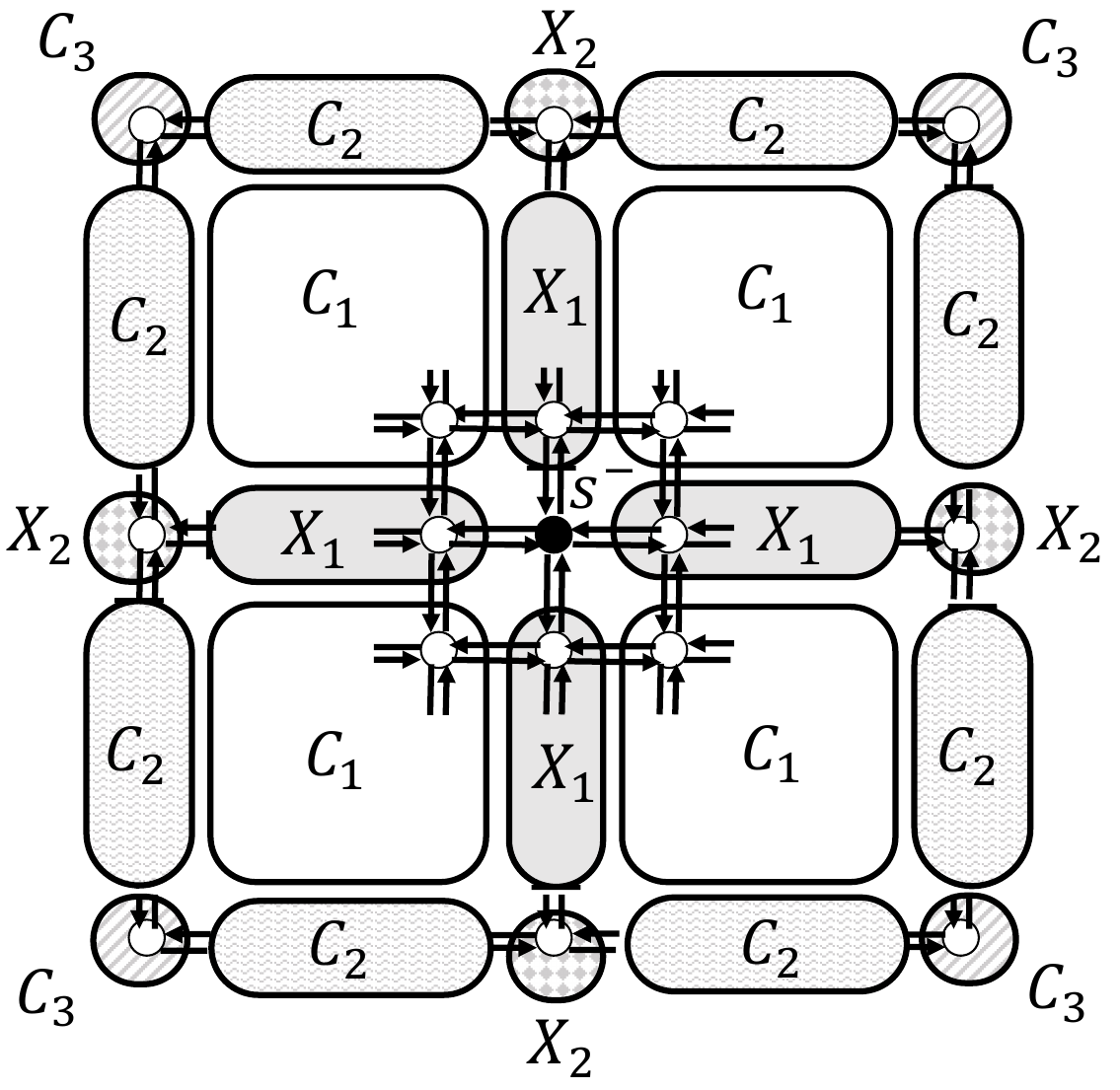}
        \caption{Areas of source subsets $C_1$, $C_2$, $C_3$, $X_1$, and $X_2$.}\label{fig:grid_graph_area_2}
    \end{minipage}
\end{tabular}
\end{figure}

For the analysis, 
we divide node set $V \setminus {s^-}$ into five areas $C_1,C_2,C_3,X_1,$ and $X_2$ as in Fig.~\ref{fig:grid_graph_area_2}. 
Let us define $\chi:=\{C_1,C_2,C_3,X_1,X_2\}$.
For any nodes $u, v \in V$, let $|uv|$ denote the minimum cost of a path from $u$ to $v$ on some static network. 
If a source subset $A$ admits $(v_1, \ldots, v_p)$, then 
$|v_1 s^-| \leq |v_i s^-|$ holds for any $i \in \{1, \ldots, p\}$
by the behavior of Algorithm~\ref{algo:successive_path}.

\begin{lemma}\label{lemm:necessary_exist_A}
For a bidirected grid network ${\cal G}$, $p\in \{1,2,3,4\}$, and $(v_1, \ldots ,v_p)\in (S^+)^p$, 
let us assume that there is a source subset admitting $(v_1, \ldots ,v_p)$. 
Then $(v_1, \ldots ,v_p)$ satisfy the following conditions:
\begin{enumerate}[(i)]
    \item if $v_1\in C_1$, then $v_1 = v_2$, $|v_3s^-|\leq |v_1s^-|+4\tau$,
    and $|v_4s^-|\leq |v_1s^-|+4\tau$ hold;
    \item if $v_1\in C_2$, then $v_1 = v_2$ and $|v_3s^-| \leq |v_1s^-| + 4\tau$ hold;
    \item if $v_1\in C_3$, then $v_1=v_2$ holds;
    \item  if $v_1\in X_1$, then $|v_2s^-|\leq |v_1s^-|+2\tau$, $|v_3s^-|\leq |v_1s^-|+2\tau$, and $|v_4s^-|\leq |v_1s^-|+8\tau$ hold;
    \item  if $v_1\in X_2$, then $|v_2s^-|\leq |v_1s^-|+2\tau$ and $|v_3s^-|\leq |v_1s^-|+2\tau$ hold.
\end{enumerate}
\end{lemma}

\begin{proof}
First, for any $v \in S^+$, there are at least two (edge-)disjoint paths from $v$ to $s^-$ in a grid graph $G$.
Thus, no subset admits any 1-tuple $(v_1) \in (S^+)^1$.
We consider the case of $p \geq 2$ and give the proof for case (i), where $v_1 \in C_1$ holds 
(the other cases can be proved in a similar way).
There are four paths from $v$ to $s^-$ in a grid graph $G$ (see Fig.~\ref{fig:C1_path}).
Thus, there exist subsets admitting $(v_1, \ldots, v_p)$ only if $p = 4$. 

In the rest of the proof, we give the necessary condition that there exists a subset admitting 
$(v_1, v_2, v_3, v_4) \in (S^+)^4$ with $v_1 \in C_1$.
By Lemma~\ref{lemm:exist_A}, it is enough to consider 
the necessary condition such that $\{v_1,v_2,v_3,v_4\}$ admits $(v_1, v_2, v_3, v_4)$.
Let $A'$ denote $\{v_1,v_2,v_3,v_4\}$. 

First, we show that $v_1 = v_2$ holds. 
There are two (edge-)disjoint shortest paths from $v_1$ to $s^-$ in $\overline{\mathcal{G}}$ (see Fig.~\ref{fig:C1_path}). 
This means that if the origin of $P^{A'}_1$ is $v_1$, then 
Algorithm~\ref{algo:successive_path} chooses the shortest path $P^{A'}_2$ with the origin $v_1$.
Thus, we have $v_1 = v_2$.

Next, we prove that $|v_3s^-|\leq |v_1s^-| + 4\tau$ holds. 
In the residual network of ${\cal G}$ constructed by $P^{A'}_1$ and $P^{A'}_2$, 
there exists a path from $v_1$ to $s^-$ of length $|v_1s^-| + 4\tau$;
see path $P_3$ in Fig.~\ref{fig:C1_path}.
Therefore, the distance from $v_3$ to $s^-$ should be less than or equal to $|v_1s^-| + 4\tau$, 
that is, $|v_3s^-| \leq |v_1s^-| + 4\tau$ holds.
In the same way, we can prove that $|v_4s^-|\leq |v_1s^-| + 4\tau$ holds 
because there exists another path from $v_1$ to $s^-$ of cost $|v_1s^-| + 4\tau$; see path $P_4$ in Fig.~\ref{fig:C1_path}.
\end{proof}

\begin{figure}[tb]
\centering
\includegraphics[width=30mm]{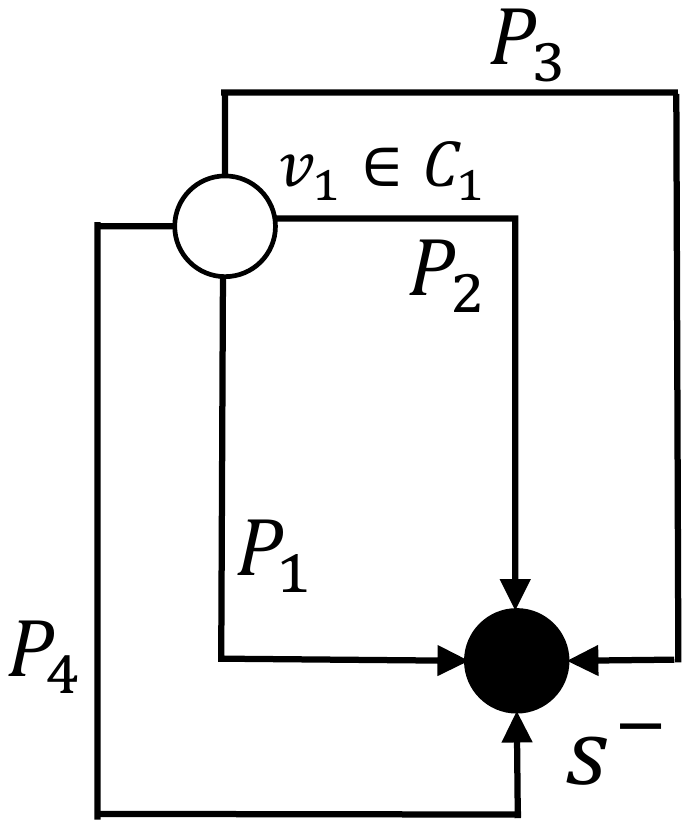}
\caption{Four paths from $v_1$ to $s^-$ if $v_1 \in C_1$ holds.}\label{fig:C1_path}
\end{figure}

For $p \in \{2,3,4\}$ and $C \in \chi$, 
let $I^p_{C}$ denote the sets of $(v_1,\ldots, v_p)$ with $v_1 \in C$ satisfying one of cases (i)--(v). 
Letting $I$ be the union of all sets $I^p_{C}$, we have the following lemma 
implying that the size of $\mathcal{\hat{A}} = O(n^2)$ because $|\mathcal{\hat{A}}| \leq |I|$ holds.

\begin{lemma}\label{lemm:I_num}
Given a bidirected grid network ${\cal G}$, $|I| = O(n^2)$ holds.
\end{lemma}
\begin{proof}
We show that $|I^4_{C_1}| = O(n^2)$ holds because $|I^p_{C_1}| = 0$ for $p = 2, 3$ by the proof of Lemma~\ref{lemm:necessary_exist_A}.
Let us assume that $v_1 \in C_1$ holds.
We estimate the number of candidates of $v_2$, $v_3$, and $v_4$
such that $(v_1, \ldots, v_4) \in I^4_{C_1}$ holds.
By Lemma~\ref{lemm:necessary_exist_A}, we have $v_2 = v_1$. 
The number of candidates of $v_3$ is $O(N) = O(\sqrt{n})$
because $v_3$ satisfies $|v_1s^-| \leq |v_3s^-| \leq |v_1s^-| + 4\tau$.
In the same way, the number of candidates of $v_4$ is also $O(\sqrt{n})$.
Therefore, $|I^4_{C_1}| = O(n) \times 1 \times O(\sqrt{n}) \times O(\sqrt{n}) = O(n^2)$
because $|C_1| = O(n)$.

In a similar way, one can show that $|I^p_{C}| = O(n^2)$ for other $C \in \{C_2, X_1, X_2\}$
and $|I^p_{C}| = O(1)$.
\end{proof}
\section{Concluding Remarks}\label{sec:quickest_p}
In this paper, we propose efficient algorithms for the evacuation problem without calling an algorithm for the submodular function minimization. 
The proposed algorithms can be generalized for solving the quickest transshipment problem in which we consider the dynamic network contains multiple sinks. 
Let $\ell$ be the number of sinks, 
$d'$ be the sum of the in-degrees of all sinks. 
In generalization, we need to consider the following two things:
(1) For computing $\theta(A)$, we distinguish $2^\ell$ cases: ``which sinks are contained in a terminal subset $A$."
(2) The number of paths obtained by the successively shortest path algorithm is at most $d'$. 
In this case, the size of $\mathcal{\hat{A}}$ is at most $2^{\ell} k^{d'}$
instead of $k^d$. 
Thus the running time of the generalized algorithm is 
$\tilde{\mathcal{O}}\left(m n d' k^{d'} 2^\ell + k^2m^2 \right)$, 
which is exponentially large on the number of sinks. 
If we are given a dynamic network that forms a grid-like network, and contains many sources and two sinks,
that is, $m = O(n)$, $k = O(n)$ and $d' \leq 8$, 
then its running time is $\tilde{\mathcal{O}}\left(n^{10}\right)$, 
which does not compete $\tilde{\mathcal{O}}\left(n^7\right)$ time algorithms by Kamiyama~\cite{Kamiyama2019} and Schl\"{o}ter~\cite{Schloter2018}.

Therefore, it is a future natural work to construct 
a more efficient algorithm for the quickest transshipment problems.



%
%
%
%

\end{document}